\documentclass[a4paper,10pt]{amsart}
\usepackage[utf8]{inputenc}
\usepackage{amsmath}
\usepackage{amsfonts}
\usepackage{amsthm}
\usepackage{amssymb}

\newcommand{\p}[1]{\mathbb{P}}

\newcommand{\norm}[2]{\Vert #1 \Vert_{#2}}
\newcommand{\fa}[1]{\forall}

\newcommand{\cN}{\mathcal{N}}

\newcommand{\la}{\lambda}
\newcommand{\Ga}{\Gamma}
\newcommand{\Si}{\Sigma}
\newcommand{\bP}{\mathbb{P}}
\newcommand{\cA}{\mathcal{A}}

\newcommand{\ti}{\times}

\newcommand{\si}{\sigma}

\newcommand{\R}{\mathbb{R}}

\newcommand{\sign}{\operatorname{sign}}

\newcommand{\Z}{\mathbb{Z}}
\newcommand{\C}{\mathbb{C}}

\theoremstyle{definition}
\newtheorem{defi}{Definition}[section]

\theoremstyle{remark}

\theoremstyle{plain}
\newtheorem{lemma}[defi]{Lemma}
\newtheorem{corollary}[defi]{Corollary}

\newtheorem{theorem}[defi]{Theorem}
\newtheorem{remark}[defi]{Remark}

\allowdisplaybreaks

\newcommand{\eps}{\varepsilon}

\newcommand{\del}{\delta}

\newcommand{\E}{\mathbb{E}}

\title[One-bit compressed sensing with partial Gaussian circulant matrices]{One-bit compressed sensing with partial Gaussian circulant matrices}
\author{Sjoerd Dirksen$^*$}
\thanks{$^*$ RWTH Aachen University, Lehrstuhl C f{\"u}r 
Mathematik (Analysis), Pontdriesch 10, 52062 Aachen,  Germany, 
lastname@mathc.rwth-aachen.de}
\author{Hans Christian Jung$^*$}
\author{Holger Rauhut$^*$} 

\date{October 9, 2017}
\begin{document}
\maketitle
\begin{abstract}
In this paper we consider memoryless one-bit compressed sensing with randomly subsampled Gaussian circulant matrices. We show that in a small sparsity regime and for small enough accuracy $\delta$, 
$m\sim \del^{-4} s\log(N/s\del)$
measurements suffice to reconstruct the direction of any $s$-sparse vector up to accuracy $\del$ via an efficient program. We derive this result by proving that partial Gaussian circulant matrices satisfy an $\ell_1/\ell_2$ RIP-property. Under a slightly worse dependence on $\del$, we establish stability with respect to approximate sparsity, as well as full vector recovery results.  
\end{abstract}

\section{Introduction}

In the past decade, compressed sensing has established itself as a new paradigm in signal processing. It predicts that one can reconstruct signals from a small number of linear measurements using efficient algorithms, by exploiting that many real-world signals possess a sparse representation. In the traditional compressed sensing literature, it is typically assumed that one can reconstruct a signal based on its analog linear measurements. In a realistic sensing scenario, measurements need to be quantized to a finite number of bits before they can be transmitted, stored, and processed. Formally, this means that one needs to reconstruct a sparse signal $x$ based on \emph{non-linear} measurements of the form $y=Q(Ax)$, where $Q:\R^m\rightarrow \cA^m$ is a quantizer and $\cA$ denotes a finite quantization alphabet.\par 
In this paper, we study the measurement model 
\begin{equation}
\label{eqn:1-bitModel}
y=\text{sign}(Ax + \tau),
\end{equation}
where $A\in\R^{m\ti N}$, $m\ll N$, $\text{sign}$ is the signum function applied element-wise and $\tau\in \R^m$ is a (possibly random) vector consisting of thresholds. Thus, every linear measurement is quantized to a single bit in a memoryless fashion, i.e., each measurement is quantized independently. This quantizer is attractive from a practical point of view, as it can be implemented using an energy-efficient comparator to a fixed voltage level (if $\tau_i=c$ for all $i$) combined with dithering (if $\tau$ is random). In the case $\tau=0$, this model was coined \emph{one-bit compressed sensing} by Boufounos and Baraniuk \cite{BoB08}. Taking all thresholds equal to zero has the disadvantage that the energy $\|x\|_2^2$ of the original signal is lost during quantization and one can only hope to recover the direction of the signal. More recent works \cite{KSW16,BFN17} have shown that by using random thresholds, it is under appropriate circumstances possible to completely reconstruct the signal (up to any prescribed precision).\par
Until now, recovery results for the one-bit compressed sensing model \eqref{eqn:1-bitModel} dealt almost exclusively with a Gaussian measurement matrix $A$. The only exception seems to be \cite{ALP14}, which deals with subgaussian matrices.  
The goal of this paper is to derive reconstruction guarantees in the case that $A$ is a randomly subsampled Gaussian circulant matrix. This compressed sensing model is important for several real-world applications, including SAR radar imaging, Fourier optical imaging and channel estimation (see e.g.\ \cite{Rom09} and the references therein). Our work seems to be the first to give rigorous reconstruction guarantees for memoryless one-bit compressed sensing involving a structured random matrix.\par 
Our results concern guarantees for uniform recovery under a \emph{small sparsity assumption}. Concretely, for a desired accuracy parameter $0 < \varepsilon \leq 1$, we assume that the sparsity $s$ is small enough, i.e.,
\[
s \lesssim \sqrt{ \varepsilon N/\log(N)}.
\]
we suppose that the (expected) number of measurements satisfy 
\begin{equation}\label{eqn:condmNint}
m \gtrsim \left\{\begin{array}{ll}\varepsilon^{-1} s \log(eN/(s\varepsilon)) & \mbox{ if } 0 < \varepsilon \leq \left(\log^2(s)\log(N)\right)^{-1}\\  
\varepsilon^{-1/2} s \log(s) \log^{3/2}(N) & \mbox{ if } \left(\log^2(s)\log(N)\right)^{-1} < \varepsilon \leq 1.
\end{array} \right.
\end{equation}
Let us first phrase our results for $\tau=0$. We consider two different recovery methods to reconstruct $x$, namely via a single hard thresholding step
\begin{equation}
\label{eqn:RHT}
\tag{HT}
x_{\text{HT}}^{\#}=H_s(A^*\sign(Ax))
\end{equation}
and via the program 
\begin{equation}
\label{eqn:LPPV}
\tag{LP}
\min_{z\in \R^n} \|z\|_1 \qquad \text{s.t.} \qquad \sign(Az)=\sign(Ax) \qquad \text{and} \qquad \|Az\|_1=1. 
\end{equation}
As the first constraint is equivalent to $(Az)_i \sign((Ax)_i)\geq 0$ for $i=1,\ldots,m$ and due to the second constraint it can be written as $\sum_{i=1}^n \sign((Ax)_i) (Az)_i = 1$, it follows that \eqref{eqn:LPPV} is a linear program.\par
Our first result shows that under \eqref{eqn:condmNint} the following holds with high probability: for any $s$-sparse $x$ with $\|x\|_2=1$ the hard thresholding reconstruction $x_{\text{HT}}^{\#}$ satisfies $\|x-x_{\text{HT}}^{\#}\|_2\leq \varepsilon^{1/4}$. Moreover, under slightly stronger conditions (see Theorem~\ref{thm:main2} with $\delta = \varepsilon^{1/4}$) for any vector satisfying $\|x\|_1\leq \sqrt{s}$ and $\|x\|_2$=1, any solution $x^{\#}_{\text{LP}}$ satisfies $\|x-x^{\#}_{\text{LP}}\|_2\leq \varepsilon^{1/8}$. As a consequence, we can reconstruct the direction $x/\|x\|_2$ of any $s$-sparse (resp.\ effectively sparse) signal via an efficient program.\par
Our second result gives guarantees for the full recovery of effectively sparse vectors, provided that an upper bound $R$ on their energy is known. We suppose that $\tau$ is a vector of independent, $\mathcal{N}(0,\pi R^2/2)$-distributed random variables. If a condition similar to \eqref{eqn:condmNint} is satisfied (see Theorem~\ref{thm:main2}), then the following holds with high probability: for any $x\in \R^N$ with $\|x\|_1\leq \sqrt{s}\|x\|_2$ and $\|x\|_2\leq R$, the solution $x^{\#}_{\text{CP}}$ to the second-order cone program 
\begin{equation}
\label{eqn:CPTh}
\tag{CP}
\min_{z \in \R^N} \|z\|_1 \qquad \text{s.t.} \qquad \text{sign}(Az+\tau) = \text{sign}(Ax+\tau), \  \|z\|_2\leq R
\end{equation}
satisfies $\|x-x_{\text{CP}}^{\#}\|_2\leq R\varepsilon^{1/8}$.\par
Our analysis relies on an observation of Foucart \cite{Fou17}, who showed that it is sufficient for the matrix $A$ to satisfy an \emph{$\ell_1/\ell_2$-RIP property} to guarantee successful uniform recovery via \eqref{eqn:RHT} and \eqref{eqn:LPPV}. In the same vein, we show that the program \eqref{eqn:CPTh} is guaranteed to succeed under an $\ell_1/\ell_2$-RIP property for a modification of $A$. We prove the required RIP-properties in Theorem~\ref{thm:mainNT}. The final section of the paper discusses two additional consequences of these RIP-results. In Corollary~\ref{cor:BPDN1} we follow the work \cite{DLR16} to derive a recovery guarantee for (unquantized) outlier-robust compressed sensing with Gaussian circulant matrices. In Theorem~\ref{thm:USC} we use a recent result from \cite{JaC16} to derive an improved guarantee for recovery from uniform scalar quantized Gaussian circulant measurements.  

\section{Related work}
\label{sec:relatedwork}

{\bf Standard compressive sensing with partial circulant matrices.}
In standard (unquantized) compressive sensing, the task is to recover an (approximately) sparse vector $x \in \R^N$ from measurements
$y = A x$, where $A \in \R^{m \times N}$ with $m \ll N$. A number of reconstruction algorithms have been introduced, most
notably $\ell_1$-minimization which computes the minimizer of 
\[
\min_{z \in \R^N} \|z\|_1 \quad \mbox{ subject to } A z = Ax.
\]
The ($\ell_2$-)restricted isometry property is a classical way of analyzing the performance of compressive sensing \cite{fora13}.
The restricted isometry constant $\delta_s$ is defined as the smallest constant $\delta$ such that
\begin{equation}\label{def:RIP2}
(1-\delta) \|x\|_2^2 \leq \|A x\|_2^2 \leq (1+ \delta)\|x\|_2^2 \quad \mbox{ for all } s\mbox{-sparse } x \in \R^N.
\end{equation}
If $\delta_{2s} < 1/\sqrt{2}$ then all $s$-sparse signals can be reconstructed via $\ell_1$-minimization exactly, see e.g.\ \cite{cazh14,fora13}. Stability under noise and sparsity defect can be shown as well and similar guarantees also hold for other reconstruction algorithms \cite{fora13}. It is well-known that Gaussian random matrices satisfy $\delta_s \leq \delta$ with probability at least $1-\eta$ if $m \gtrsim \delta^{-2} (s \log(eN/s) + \log(1/\eta))$ \cite[Chapter 9]{fora13}.

The situation that $A$ is a subsampled random circulant matrix (see below for a formal definition) has been analyzed in
several contributions \cite{Rom09,ra09,ra10,KMR14,jara15,MRW16}. The best available result states \cite{KMR14} 
that a properly normalized (deterministically) subsampled random
circulant matrix (generated by a Gaussian random vector) satisfies $\delta_s \leq \delta$ with probability at least
$1-\eta$ if
\[
m \gtrsim \delta^{-2} s (\log^2(s) \log^2(N) + \log(1/\eta)).
\]
The original contribution \cite{Rom09} by Romberg uses random subsampling of a circulant matrix and requires slightly more logarithmic factors, but is able to treat sparsity with respect to an arbitrary basis. In the case of randomly subsampled random convolutions and sparsity with respect to the standard basis, stable and robust $s$-sparse recovery via $\ell_1$-minimization could recently
be shown via the null space property \cite{fora13} in \cite{MRW16} in a small sparsity regime $s \lesssim \sqrt{N/\log(N)}$
under the optimal condition
\begin{equation}\label{m:optimal}
m \gtrsim s \log(eN/s).
\end{equation}
We note that the proof in \cite{MRW16} provides the lower RIP-bound in \eqref{def:RIP2} and may be
extended to show the lower $\ell_1/\ell_2$ RIP bound in \eqref{eqn:RIP12} below under condition \eqref{m:optimal}.

Non-uniform recovery results
have been shown in \cite{ra09,ra10,jara15} which require only $m \gtrsim s \log(N)$ measurements 
for exact recovery from (deterministically) subsampled random convolutions via $\ell_1$-minimization.\par

{\bf One-bit compressive sensing with Gaussian measurements, $\mathbf{\tau=0}$.} The majority of the known signal reconstruction results in one-bit compressed sensing are restricted to standard Gaussian measurement matrices. Let us first consider the results in the case $\tau=0$. It was shown in \cite[Theorem 2]{JLB13} that if $A$ is $m\times N$ Gaussian and $m\gtrsim \del^{-1}s\log(N/\del)$ then, with high probability, any $s$-sparse $x,x'$ with $\|x\|_2=\|x'\|_2=1$ and 
$\sign(Ax)=\sign(Ax')$ satisfy $\|x-x'\|_2\leq \del$. In particular, this shows that one can approximate 
$x$ up to error $\del$ by the solution of the non-convex program
$$\min\|z\|_0 \qquad \text{s.t.} \qquad \sign(Ax)=\sign(Az), \ \|z\|_2=1.$$
This result is near-optimal in the following sense: any reconstruction $x^{\#}$ based on $\sign(Ax)$ satisfies $\|x^{\#}-x\|_2\gtrsim s/(m+s^{3/2})$ \cite[Theorem 1]{JLB13}. That is, the dependence of $m$ on $\del$ can in general not be improved. It was shown in \cite[Theorem 7]{GNJ13} that this optimal error dependence can be obtained using a polynomial time algorithm if the measurement matrix is modified. Specifically, they showed that if $m\gtrsim \del^{-1} m' \log(m'/\del)$ if $A=A_2A_1$, where $A_2$ is $m\times m'$ Gaussian and $A_1$ is any $m'\times N$ matrix with RIP constant bounded by $1/6$ (so one can take $m'\sim s\log(N/s)$ if $A_1$ is Gaussian), then with high probability one can recover any $s$-sparse $x$ with unit norm up to error $\del$ from $\sign(Ax)$ using an efficient algorithm. To recover efficiently from Gaussian one-bit measurements, Plan and Vershynin \cite{PlV13lin} proposed the reconstruction program in \eqref{eqn:LPPV}. They showed that using $m \gtrsim \del^{-1} s\log^2(N/s)$ Gaussian measurements one can recover every $x$ with $\|x\|_1\leq\sqrt{s}$ and $\|x\|_2=1$ via \eqref{eqn:LPPV} with reconstruction error $\del^{1/5}$. In \cite{PlV13} they introduced a different convex program and showed that if $m\gtrsim \del^{-1}s\log(N/s)$, then one can achieve a reconstruction error $\del^{1/6}$ even if there is (adversarial) quantization noise present.\par
{\bf Thresholds.} It was recently shown that one can recover full signals (instead of just their directions) by incorporating appropriate thresholds. In \cite{KSW16} it was shown that by taking Gaussian thresholds $\tau_i$ one can recover energy information by slightly modifying the linear program \eqref{eqn:LPPV}. A similar observation for recovery using the program \eqref{eqn:CPTh} was made in \cite{BFN17}. The paper \cite{KSW16} also proposed a method to estimate $\|x\|_2$ using a single deterministic threshold $\tau_i=\tau$ that works well if one has some prior knowledge of the energy range. \par 
{\bf Subgaussian measurements.} The results described above are all restricted to Gaussian measurements. It seems that \cite{ALP14} is currently the only work on memoryless one-bit compressed sensing for non-Gaussian matrices. Even though one-bit compressed sensing can fail in general for subgaussian matrices, it is shown in \cite{ALP14} that some non-uniform recovery results from \cite{PlV13} can be extended to subgaussian matrices if the signal to be recovered is not too sparse (meaning that $\|x\|_{\infty}$ is small) or if the measurement vectors are close to Gaussian in terms of the total variation distance.\par 
{\bf Uniform scalar quantization.}  Some recovery results for circulant matrices are essentially known for a different memoryless quantization scheme. Consider the uniform scalar quantizer $Q_{\del}:\R^m\rightarrow (\del\Z + \del/2)^m$ defined by $Q_{\del}(z) = \big(\del\lfloor z_i/\del\rfloor + \del/2\big)_{i=1}^m$. As we point out in Appendix \ref{app:scalar-quant}, if $A$ consists of $m\gtrsim s\log^2 s\log^2 N$ deterministic samples of a subgaussian circulant matrix, then it follows from \cite{KMR14} that with high probability one can recover any $s$-sparse vector up to a reconstruction error $\del$ from its quantized measurements $Q_{\del}(Ax)$. By using random subsampling and imposing a small sparsity assumption similar to ours, this number of measurements can be decreased to $m\gtrsim s\log(N/s)$ \cite{MRW16}. In these results, the recovery error does not improve beyond the resolution $\del$ of the quantizer, even if one takes more measurements. In Theorem~\ref{thm:USC} we show that for a randomly subsampled Gaussian circulant matrix it is possible to achieve a reconstruction error decay beyond the quantization resolution, provided that one introduces an appropriate dithering in the quantizer.\par
{\bf Adaptive quantization methods.} The results discussed above all concern \emph{memoryless} quantization schemes, meaning that each measurement is quantized independently. By quantizing adaptively based on previous measurements, one can improve the reconstruction error decay rate. In \cite{BFN17} it was shown for Gaussian measurement matrices that by using adaptive thresholds, one can even achieve an (optimal) exponential decay in terms of the number of measurements. Very recently, it was shown that one can efficiently recover a signal from randomly subsampled subgaussian partial circulant measurements that have been quantized using a popular scheme called sigma-delta quantization \cite{FKS17}. In particular, \cite[Theorem 5]{FKS17} proves that based on $m\sim s\log^2 s\log^2 N$ one-bit sigma-delta quantized measurements, one can use a convex program to find an approximant of the signal that exhibits polynomial reconstruction error decay. Although adaptive methods such as sigma-delta quantization can achieve a better error decay than memoryless quantization schemes, they require a more complicated hardware architecture and higher energy consumption in operation than the memoryless quantizer studied in this work.

\section{Notation}
\label{sec:notation}

We use $\text{Id}_N$ to denote the $N\ti N$ identity matrix. If $A\in \R^{m\ti N}$ and $B\in \R^{m\ti M}$, then $[A \ B] \in \R^{m\ti (N+M)}$ is the matrix obtained by concatenating $A$ and $B$. For $x\in \R^N$ and $y\in \R^M$ let $[x,y]\in \R^{N+M}$ be vector obtained by appending $y$ at the end of $x$. If $I\subset[N]$, then $x_I\in \R^{|I|}$ is the vector obtained by restricting $x$ to its entries in $I$. We let $R_I:\R^N\to \R^{|I|}$, $R_I(x)=x_I$ denote the restriction operator.
Further, $\|x\|_p$ denotes the usual $\ell_p$-norm of a vector $x$, $\|A\|_{\ell_2 \to \ell_2}$ the spectral norm of a matrix $A$ and $\|A\|_F$ its Frobenius norm. For an event $E$, $1_E$ is the characteristic function of $E$.
\par  
We let $\Si_{s,N}$ denote the set of all $s$-sparse vectors with unit norm. We say that $x\in \R^N$ is \emph{$s$-effectively sparse} if $\|x\|_1\leq \sqrt{s}\|x\|_2$. We let $\Sigma_{s,N}^\text{eff}$ denote the set of all $s$-effectively sparse vectors. Clearly, if $x$ is $s$-sparse, then it is $s$-effectively sparse. We let $H_s$ denote the hard thresholding operator, which sets all coefficients of a vector except the $s$ largest ones (in absolute value) to $0$.\par 
For any $x\in \R^N$ we let $\Ga_x \in \R^{N\ti N}$ and $D_x \in \R^N$ be the circulant matrix and diagonal matrix, respectively, generated by $x$. That is,
$$D_x = \begin{bmatrix}
x_1 & 0 & \cdots & 0 & 0\\
0 & x_1 & 0 & & 0\\
0 & 0 & x_2 & \ddots & \vdots\\
\vdots & &  \ddots & \ddots & 0\\
0 & \cdots & & 0 & x_{N}
\end{bmatrix},\ \Gamma_x=
\begin{bmatrix}
x_{N} & x_{1} & x_2 & \cdots & x_{N-2} & x_{N-1}\\
x_{N-1} & x_N & x_1 &  \cdots & x_{N-3} & x_{N-2}\\
x_{N-2} & x_{N-1} & x_N & \cdots & x_{N-4} & x_{N-3}\\
\vdots & \vdots & \vdots&  \vdots & \vdots & \vdots\\
x_1 & x_2 & x_3 & \cdots & x_{N-1} & x_N
\end{bmatrix} \; .
$$ 
We study the following linear measurement matrix. We consider a vector $\theta$ of i.i.d.\ random selectors with mean $m/N$ and let $I=\{i\in [N] \ : \ \theta_i=1\}$. Let $g$ be an $N$-dimensional standard Gaussian vector that is independent of $\theta$. We define the \emph{randomly subsampled Gaussian circulant matrix} by $A=R_I\Gamma_g$. Note that $\E|I|=m$, so $m$ corresponds to the expected number of measurements in this model.

\section{Recovery via RIP$_{1,2}$-properties}

Let us start by stating our main recovery result for vectors with small sparsity located on the unit sphere.
\begin{theorem}\label{thm:main1}
\label{thm:dirRecov}
Let $0 < \delta, \eta \leq 1$ and $s \in [N]$ 
such that
\begin{equation}\label{s:cond}
s \lesssim \min\{\sqrt{ \delta^2 N/\log(N)}, \delta^2 N/\log(1/\eta)\}.
\end{equation}
If $0 < \delta \leq 
\big( \log(s) \sqrt{\log(N)}\big)^{-1}$
suppose that
\begin{equation}
\label{eqn:condSp:small}
m \gtrsim \del^{-2} s \log(eN/(s\del \eta)). 
\end{equation}
If $\big( \log(s) \sqrt{\log(N)}\big)^{-1} < \delta \leq 1$
suppose that
\begin{equation}\label{eqn:condSp:large}
m \gtrsim \delta^{-1} s \max\left\{ \log(s) \log^{3/2}(N),  \frac{\log(1/\eta)}{\log(s) \sqrt{\log(N)}}, \frac{\log(1/\eta) \log(s) \sqrt{\log(N)}}{s}\right\}.
\end{equation}
Let $A=R_I\Gamma_g$. Then, with probability at least $1-\eta$, for every $x\in \R^N$ with $\|x\|_0\leq s$ and $\|x\|_2=1$,  the hard thresholding reconstruction $x_{\text{HT}}^{\#}$ satisfies $\|x-x_{\text{HT}}^{\#}\|_2\lesssim \sqrt{\del}$. 
\end{theorem}
Let us remark that for polynomially scaling probabilities $\eta = N^{-\alpha}$ the second and third term in the maximum in \eqref{eqn:condSp:large} can be bounded by a constant $c_\alpha$ times the first term and then $\eqref{eqn:condSp:large}$
reduces to
\[
m \gtrsim \delta^{-1} s \log(s) \log^{3/2}(N),
\]
which is implied by the even simpler condition $m \gtrsim \delta^{-1} s \log^{5/2}(N)$.
We further note that \eqref{s:cond} also gives an implicit condition on $\delta$. In fact, if $\delta \lesssim \sqrt{\log(N)/N}$, then \eqref{s:cond} excludes all nontrivial sparsities $s \geq 1$. However, we expect that the requirement \eqref{s:cond} of small sparsity is only an artefact of our proof and that recovery can also be expected for larger sparsities under conditions
similar to \eqref{eqn:condSp:small} and \eqref{eqn:condSp:large} with possibly more logarithmic factors, see also \cite{MRW16} for an analogue phenomenon for standard compressed sensing. In fact, our proof relies on the RIP$_{1,2}$ and \cite{MRW16} provides at least the lower RIP$_{1,2}$ bound also for larger sparsities.

Let us now state our main result for the LP-reconstruction which unfortunately requires worse scaling in $\delta$ than \eqref{eqn:condSp:small} and \eqref{eqn:condSp:large}.
\begin{theorem}\label{thm:main2} Let $0 < \delta, \eta \leq 1$. 
If $0 < \delta \leq (\log^2(s) \log(N))^{-1/4}$ assume that
\begin{align*}
s & \lesssim \min \{ \sqrt{\delta^4 N/\log(N)}, \delta^2 N/\log(1/\eta) \},\\
m & \gtrsim \delta^{-4} s \log(eN/s),
\end{align*}
and if $(\log^2(s) \log(N))^{-1/4} < \delta \leq 1$ assume that
\begin{align*}
s & \lesssim \min\left\{\delta^{4/3}\sqrt{N/\log^2(N)}, \delta^2 N/\log(1/\eta)\right\}\\
m & \gtrsim \delta^{-4/3} s \max\left\{\log^{4/3}(s) \log^{5/3}(N), \frac{\log(1/\eta)}{\log^{3/2}(s) \log^{1/3}(N)}, \frac{\log^{2/3}(s) \log^{1/3}(N)}{\delta^{4/3} s}\right\}.
\end{align*}
Then the following holds with probability exceeding $1-\eta$: for every $x\in \R^N$ with $\|x\|_1\leq s$ and $\|x\|_2=1$,  the LP-reconstruction $x_{\text{LP}}^{\#}$ satisfies
$\|x-x_{\text{LP}}^{\#}\|_2\lesssim \sqrt{\del}$.
\end{theorem}
We will prove Theorem~\ref{thm:dirRecov} by using a recent observation of Foucart \cite{Fou17}. He showed that one can accurately recover signals from one-bit measurements if the measurement matrix satisfies an appropriate RIP-type property. Let us say that a matrix $A$ \emph{satisfies $\operatorname{RIP}_{1,2}(s,\del)$} if
\begin{equation}\label{eqn:RIP12}
(1-\del)\|x\|_2 \leq \|Ax\|_1 \leq (1+\del)\|x\|_2,\qquad \text{for all} \ x\in \Si_{s,N}
\end{equation}
and \emph{$A$ satisfies $\operatorname{RIP}^\text{eff}_{1,2}(s,\del)$} if
\begin{equation}
\label{eqn:RIPeff}
(1-\del)\|x\|_2 \leq \|Ax\|_1 \leq (1+\del)\|x\|_2,\qquad \text{for all} \ x\in \Si_{s,N}^{\text{eff}}.
\end{equation}
\begin{lemma}
\label{lem:FouHT}
\cite[Theorem 8]{Fou17} 
Suppose that $A$ satisfies $\operatorname{RIP}_{1,2}(2s,\del)$. Then, for every $x\in \R^N$ with $\|x\|_0\leq s$ and $\|x\|_2=1$,  the hard thresholding reconstruction $x_{\text{HT}}^{\#}$ satisfies $\|x-x_{\text{HT}}^{\#}\|_2\leq 2\sqrt{5\del}$.\par
Let $\del\leq 1/5$. Suppose that $A$ satisfies $\operatorname{RIP}_{1,2}^{\text{eff}}(9s,\del)$. Then, for every $x\in \R^N$ with $\|x\|_1\leq s$ and $\|x\|_2=1$,  the LP-reconstruction $x_{\text{LP}}^{\#}$ satisfies
$\|x-x_{\text{LP}}^{\#}\|_2\leq 2\sqrt{5\del}$. 
\end{lemma}
\begin{remark}
It is in general not possible to extend Theorem~\ref{thm:dirRecov} to subgaussian circulant matrices. Indeed, suppose that the measurement matrix is a (rescaled, subsampled) Bernoulli circulant matrix, the threshold vector $\tau$ in \eqref{eqn:1-bitModel} is zero and consider, for $0<\lambda<1$, the normalized $2$-sparse vectors
\begin{equation}
\label{eqn:xlaDef}
x_{+\la} = (1+\la^2)^{-1/2}(1,\la,0,\ldots,0), \qquad x_{-\la} = (1+\la^2)^{-1/2}(1,-\la,0,\ldots,0).
\end{equation}
Then 
\begin{align*}
\sign(\langle (\Ga_{\eps})_i,x_{+\la}\rangle)& =\sign(\eps_{n+1-i} + \la \eps_{n+2-i}) \\
& =  \sign(\eps_{n+1-i}) = \sign(\eps_{n+1-i} - \la \eps_{n+2-i}) =  \sign(\langle (\Ga_{\eps})_i,x_{-\la}\rangle)
\end{align*}
That is, $x_{+\la}$ and $x_{-\la}$ produce identical one-bit measurements.\par  
As a consequence, Lemma~\ref{lem:FouHT} implies that a subsampled Bernoulli circulant matrix cannot satisfy the RIP$_{1,2}(4,\del)$ property for small values of $\del$, regardless of how we subsample and scale the matrix. Indeed, suppose that $A=\alpha R_I \Ga_\eps$ satisfies this property for a suitable $I\subset[N]$ and scaling factor $\alpha$. 
Since $\sign(Ax_{+\la})=\sign(Ax_{-\la})$, we find using Lemma~\ref{lem:FouHT} and the triangle inequality
\begin{align*}
\frac{2\la}{(1+\la^2)^{1/2}} & = \|x_{+\la}-x_{-\la}\|_2 \\
& \leq \|x_{+\la}-H_s(A^*\sign(Ax_{+\la}))\|_2 + \|H_s(A^*\sign(Ax_{-\la}))-x_{-\la}\|_2\\
& \leq 4\sqrt{5\del}.
\end{align*}
By taking $\la\rightarrow 1$ we find $\del\geq 1/40$.

However, by excluding extremely sparse vectors via a suitable $\ell_\infty$-norm bound, it might be possible to work around this counter example. In fact, in the case of unstructured subgaussian random matrices, positive recovery results for sparse vectors with such additional constraint were shown in \cite{ALP14}.
\end{remark}
So far, our recovery results only allow to recover vectors lying on the unit sphere. By incorporating Gaussian dithering in the quantization process we can reconstruct any effectively sparse vector, provided that we have an a-priori upper bound on its energy. 
\begin{theorem}
\label{thm:fullRecov}
Let $A=R_I\Gamma_g$ and let $\tau_1,\ldots,\tau_m$ be independent $\mathcal{N}(0,\pi R^2/2)$-distributed random variables. 
Under the assumptions on $s,m,n,\delta,\eta$ of Theorem~\ref{thm:main2} the following holds with probability exceeding $1-\eta$: for any $x\in \R^N$ with $\|x\|_1\leq \sqrt{s}\|x\|_2$ and $\|x\|_2\leq R$, any solution $x_{\text{CP}}^{\#}$ to the second-order cone program \eqref{eqn:CPTh} satisfies $\|x-x_{\text{CP}}^{\#}\|_2\leq R\sqrt{\del}$.
\end{theorem} 
To prove this result, we let $C\in\R^{m\ti (N+1)}$ and consider the following abstract version of \eqref{eqn:CPTh}: 
\begin{equation} 
\label{eqn:CPabs}
  \min_{z \in \R^N} \|z\|_1 \qquad \text{s.t.} \qquad \text{sign}(C[z,R]) = \text{sign}(C[x,R]), \  \|z\|_2\leq R.
\end{equation}
It is straightforward to verify that \eqref{eqn:CPTh} is obtained by taking $C=\frac{1}{m}\sqrt{\frac{\pi}{2}}B$, where
\begin{equation}
\label{eqn:Bdef}
B:=R_I [\Gamma_g \ h] = R_I \begin{bmatrix}
g_{N} & g_{1} & g_2 & \cdots & g_{N-2} & g_{N-1} & h_1\\
g_{N-1} & g_N & g_1 &  \cdots & g_{N-3} & g_{N-2} & h_2 \\
g_{N-2} & g_{N-1} & g_N & \cdots & g_{N-4} & g_{N-3}& h_3\\
\vdots & \vdots & \vdots&  \vdots & \vdots & \vdots & \vdots\\
g_1 & g_2 & g_3 & \cdots & g_{N-1} & g_N & h_N
\end{bmatrix},
\end{equation}
and $h$ is a standard Gaussian vector that is independent of $\theta$ and $g$.
\begin{lemma}
\label{lem:RIPfullRec}
Let $\del<1/5$. Suppose that $C$ satisfies $\text{RIP}_{1,2}^{\text{eff}}(36(\sqrt{s}+1)^2,\del)$. Then, for any $x\in \R^N$ satisfying $\|x\|_1\leq \sqrt{s}\|x\|_2$ and $\|x\|_2\leq R$, any solution $x^{\#}$ to \eqref{eqn:CPabs} satisfies
$$\|x-x^{\#}\|_2\leq 2R\sqrt{\del}.$$ 
\end{lemma}
The following proof is based on arguments in \cite[Section 8.4]{Fou17} and \cite[Corollary 9]{BFN17}.
\begin{proof}
In the proof of \cite[Corollary 9]{BFN17} it shown that 
\begin{equation}
\label{eqn:normalization}
 \norm{u-v}{2} \leq 2 \left \Vert \frac{[u,1]}{\norm{[u,1]}{2}} - \frac{[v,1]}{\norm{[v,1]}{2}} 
 \right \Vert_2 
\end{equation}
for any two vectors $u,v \in B_{\ell_2^N}$. Let $x \in  \Si_{s,N}^{\text{eff}} \cap R B_{\ell_2^N}$, 
$x^{\#}$ be any solution to \eqref{eqn:CPabs} and write
$$\bar{x} = [x,R]/\norm{[x,R]}{2}, \qquad \bar{x}^\#=[x^\#,R]/\norm{[x^\#,R]}{2}.$$
Since $x/R,x^\#/R\in B_{\ell_2^N}$, \eqref{eqn:normalization} implies that
$$\|x-x^{\#}\|_2 \leq 2R \|\bar{x}-\bar{x}^{\#}\|_2.$$ 
By the parallellogram identity,
\begin{equation}
\label{eqn:parallell}
\Big\|\frac{\bar{x}-\bar{x}^\#}{2}\Big\|_2^2= \frac{\norm{\bar{x}}{2}^2 + \norm{\bar{x}^\#}{2}^2}{2} - \left \Vert \frac{\bar{x}+\bar{x}^\#}{2} \right\Vert_2^2.
\end{equation}
Let us observe that $[x,R]$ and $[x^\#,R]$ are $(\sqrt{s}+1)^2$-effectively sparse. Indeed, by optimality of $x^{\#}$ for \eqref{eqn:CPabs} and $s$-effective sparsity of $x$,
$$\|[x^{\#},R]\|_1\leq \|[x,R]\|_1 \leq \sqrt{s}\|x\|_2 + R\leq R(\sqrt{s}+1)$$
and $\|[x,R]\|_2,\|[x^{\#},R]\|_2\geq R$. We claim that 
\begin{equation}
\label{eqn:CCeff}
z:=\frac{\bar{x}+\bar{x}^\#}{2} \in \Sigma_{36(\sqrt{s}+1)^2,N+1}^\text{eff} \; .
\end{equation}
Once this is shown, we can use $\text{sign}(C\bar{x}) = \text{sign}(C\bar{x}^\#)$ and the RIP$_{1,2}^{\text{eff}}(36(\sqrt{s}+1)^2,\del)$-property of $C$ to find
\begin{align}
\label{eqn:zNormLB}
  \left\Vert \frac{\bar{x}+\bar{x}^\#}{2} \right\Vert_2 &\geq \frac{1}{1+\del}\left\Vert C\Big(\frac{\bar{x}+\bar{x}^\#}{2}\Big)\right\Vert_1 =  
  \frac{\norm{C\bar{x}}{1} + \norm{C\bar{x}^\#}{1}}{2(1+\del)} \geq \frac{(1-\del)}{(1+\del)}.
\end{align}
Hence, \eqref{eqn:parallell} implies
$$\Big\|\frac{\bar{x}-\bar{x}^\#}{2}\Big\|_2^2\leq 1-\frac{(1-\del)^2}{(1+\del)^2}\leq \frac{4\del}{(1+\del)^2}.$$
Let us now prove \eqref{eqn:CCeff}. Since $[x,R]$ and $[x^\#,R]$ are $(\sqrt{s}+1)^2$-effectively sparse,
$$\|z\|_1 \leq \frac{1}{2} \Big\|\frac{[x,R]}{\|[x,R]\|_2}\Big\|_1 + \frac{1}{2} \Big\|\frac{[x^\#,R]}{\|[x^\#,R]\|_2}\Big\|_1\leq \sqrt{s}+1.$$
It remains to bound $\|z\|_2$ from below. In \eqref{eqn:zNormLB} we already observed that
\begin{equation}
\label{eqn:Czlower}
\|Cz\|_1 = \frac{1}{2} \Big\|\frac{C[x,R]}{\|[x,R]\|_2}\Big\|_1 + \frac{1}{2} \Big\|\frac{C[x^\#,R]}{\|[x^\#,R]\|_2}\Big\|_1 \geq (1-\del).
\end{equation}
Set $t=8s+8\geq 4(\sqrt{s}+1)^2$. Let $T_0$ be the index set corresponding to the $t$ largest entries of $z$, $T_1$ be the set corresponding to the next $t$ largest entries of $z$, and so on. Then, for all $k\geq 1$,
$$\|z_{T_k}\|_2 \leq \sqrt{t}\|z_{T_k}\|_{\infty} \leq \|z_{T_{k-1}}\|_1/\sqrt{t}.$$
Since $C$ satisfies RIP$_{1,2}^{\text{eff}}(36(\sqrt{s}+1)^2,\del)$, it satisfies RIP$_{1,2}(t,\del)$ and hence
\begin{align}
\label{eqn:Czupper}
\|Cz\|_1 & \leq \sum_{k\geq 0} \|Cz_{T_k}\|_1 \leq (1+\del)\Big(\|z_{T_0}\|_2 + \sum_{k\geq 1} \|z_{T_k}\|_2\Big) \nonumber \\
& \leq (1+\del)\|z\|_2 + \frac{(1+\del)}{\sqrt{t}}\|z\|_1 \leq (1+\del)\|z\|_2 + \frac{(1+\del)}{\sqrt{t}}(\sqrt{s}+1). \nonumber\\
& \leq (1+\del)\|z\|_2 + \frac{1}{2}(1+\del). 
\end{align}
Since $\del\leq 1/5$, \eqref{eqn:Czlower} and \eqref{eqn:Czupper} together yield 
$$\|z\|_2\geq \frac{(1-\del) - \frac{1}{2}(1+\del)}{(1+\del)} =\frac{\frac{1}{2}-\frac{3}{2}\del}{1+\del} \geq \frac{1}{6}.$$
\end{proof} 

\section{Proof of the RIP$_{1,2}$-properties}

We will now prove the RIP-properties required in Lemmas~\ref{lem:FouHT} and \ref{lem:RIPfullRec}. For our analysis we recall two standard concentration inequalities. The Hanson-Wright inequality \cite{HaW71} states that if $g$ is a standard Gaussian vector and $B\in \R^{N\ti N}$, then 
\begin{equation}
\label{eqn:H-W}
\bP(|g^T Bg-\E g^T Bg|\geq t) \leq \exp\Big(-c\min\Big\{\frac{t^2}{\|B\|_F^2},\frac{t}{\|B\|_{\ell_2\to\ell_2}}\Big\}\Big),
\end{equation}
where $c$ is an absolute constant and $\|B\|_F$ and $\|B\|_{\ell_2\to\ell_2}$ are the Frobenius and operator norms of $B$, respectively. We refer to \cite{RuV13} for a modern proof. In addition, we will use a well-known concentration inequality for suprema of Gaussian processes (see e.g.\ \cite[Theorem 5.8]{BLM13}). If $T\subset \R^N$, then
\begin{equation}
\label{eqn:supGauss}
\bP\Big(\Big|\sup_{x\in T} \langle x,g\rangle - \E\sup_{x\in T} \langle x,g\rangle\Big|\geq t\Big)\leq 2e^{-t^2/2\si^2},
\end{equation}
where $\si^2 = \sup_{x\in T} \|x\|_2^2$. We will use the following observation.
\begin{lemma}
Suppose that $y\in \R^N$ satisfies $\|y\|_1\leq \sqrt{s}$ and $\|y\|_2=1$. Let $g$ be a standard Gaussian vector. For any $t>0$,
\begin{equation}
\label{eqn:Gammal2}
\bP\Big(\Big|\frac{1}{N}\|\Gamma_g y\|_2^2 - 1\Big| \geq t\Big) \leq 2e^{-cN\min\{t^2,t\}/s}
\end{equation}
and
\begin{equation}
\label{eqn:Gammal1}
\bP\Big(\Big|\frac{1}{N} \sqrt{\frac{\pi}{2}} \|\Ga_g y\|_1 - 1\Big|\geq t\Big) \leq 2e^{-N t^2/\pi s}.
\end{equation} 
If $h$ is a standard Gaussian vector and $y\in \R^{N+1}$ satisfies $\|y\|_1\leq \sqrt{s}$ and $\|y\|_2=1$, then 
\begin{equation}
\label{eqn:Gammal1app}
\bP\Big(\Big|\frac{1}{N} \sqrt{\frac{\pi}{2}} \|[\Ga_g \ h] y\|_1 - 1\Big|\geq t\Big) \leq 2e^{-N t^2/\pi s}.
\end{equation} 
\end{lemma}
\begin{proof}
Note that $\Gamma_g y=\Gamma_y g$. By the Hanson-Wright inequality \eqref{eqn:H-W}
 \begin{equation*}
  \bP\Big(| \|\Gamma_y g\|_2^2 - N\|y\|_2^2| \geq Nt\Big) 
  \leq \exp\left(-c\min\left\{ 
  \frac{t^2 N^2}{\|\Gamma_y^* \Gamma_y\|_F^2}, 
  \frac{t N}{\|\Gamma_y^* \Gamma_y\|_{\ell^2 \to \ell^2}}.
  \right\}\right) \; .
 \end{equation*}
Recall that the convolution satisfies $\|\Gamma_y z\|_2 = \|y * z\|_2 \leq \|y\|_1 \|z\|_2 \leq \sqrt{s} \|z\|_2$ for all $z \in \R^N$,
which implies 
\begin{align}\label{Gamma:opnorm}
\|\Gamma_y^* \Gamma_y \|_{\ell_2 \to \ell_2} & = \|\Gamma_y\|^2_{\ell_2 \to \ell_2}  \leq s, \\
\|\Gamma_y^* \Gamma_y\|_F &\leq \|\Gamma_y\|_{\ell_2 \to \ell_2} \|\Gamma_y\|_F \leq \sqrt{s} \sqrt{N} \|y\|_2 = \sqrt{sN}. \notag
 \end{align}
To prove \eqref{eqn:Gammal1} observe that
$$\E\Big(\frac{1}{N}\sqrt{\frac{\pi}{2}} \|\Ga_g y\|_1\Big) = \sqrt{\frac{\pi}{2}}\E|\langle g,y\rangle| = \|y\|_2 = 1$$
and
$$\|\Ga_g y\|_1 = \sup_{z\in B_{\ell_{\infty}^N}} \langle z,\Gamma_g y\rangle  = \sup_{z\in B_{\ell_{\infty}^N}} \langle z,\Gamma_y g\rangle = \sup_{z\in B_{\ell_{\infty}^N}} \langle \Ga_y^* z,g\rangle.$$
Hence, by the concentration inequality \eqref{eqn:supGauss} for suprema of Gaussian processes, 
$$\bP\Big(\Big|\frac{1}{N} \sqrt{\frac{\pi}{2}} \|\Ga_g y\|_1 - 1\Big|\geq t\Big) \leq 2e^{-t^2/2\si_y^2},$$
where
$$\si_y^2 = 
\frac{1}{N^2}\frac{\pi}{2} \sup_{z\in B_{\ell_{\infty}^N}} \|\Ga_y^* z\|_2^2.$$
For any $z\in B_{\ell_{\infty}^N}$, we obtain using 
\eqref{Gamma:opnorm} 
$$\|\Ga_y^* z\|_2^2 \leq \|\Ga_y\|_{\ell_2 \to \ell_2}^2 \|z\|_2^2 \leq sN.$$
The proof of \eqref{eqn:Gammal1app} is similar. By writing $y=[y_{[N]},y_{N+1}]$ we find
$$[\Ga_g \ h]y=\Ga_g y_{[N]} + y_{N+1}h = [\Gamma_{y_{[N]}} \ y_{N+1}\text{Id}_N][g,h].$$
It follows that
$$\|[\Ga_g \ h]y\|_1 = \sup_{z\in B_{\ell_{\infty}^{N}}} \langle z, [\Gamma_{y_{[N]}} \ y_{N+1}\text{Id}_N][g,h]\rangle  = \sup_{z\in B_{\ell_{\infty}^{N}}} \langle  [\Gamma_{y_{[N]}} \ y_{N+1}\text{Id}_N]^*z,[g,h]\rangle.$$
Moreover,
$$\E\Big(\frac{1}{N}\sqrt{\frac{\pi}{2}} \|[\Ga_g \ h] y\|_1\Big) = \|y\|_2=1.$$
The concentration inequality \eqref{eqn:supGauss} for suprema of Gaussian processes now implies 
$$\bP\Big(\Big|\frac{1}{N} \sqrt{\frac{\pi}{2}} \|\Ga_g y\|_1 - 1\Big|\geq t\Big) \leq 2e^{-t^2/2\si_y^2},$$
where
\begin{equation*}
\si_y^2 = \frac{1}{N^2}\frac{\pi}{2} \sup_{z\in B_{\ell_{\infty}^{N}}} \|[\Gamma_{y_{[N]}} \ y_{N+1}\text{Id}_N]^*z\|_2^2.
\end{equation*}
For any $z\in B_{\ell_{\infty}^N}$, we obtain using \eqref{Gamma:opnorm} and the Cauchy-Schwarz inequality,
\begin{align*}
\|[\Gamma_{y_{[N]}} \ y_{N+1}\text{Id}_N]^*z\|_2^2 & = \|\Gamma_{y_{[N]}}^*z\|_2^2 + y_{N+1}^2 \ \|z\|_2^2 \\
& \leq (\|\Gamma_{y_{[N]}}\|_{\ell_2 \to \ell_2}^2 + y_{N+1}^2) \|z\|_2^2 \leq (s\|y_{[N]}\|_2^2 + y_{N+1}^2)N\leq sN.
\end{align*}
This completes the proof.
\end{proof}
By Lemmas~\ref{lem:FouHT} and \ref{lem:RIPfullRec}, our main recovery results in Theorems~\ref{thm:dirRecov} and \ref{thm:fullRecov} are implied by the following theorem. As before, $\theta$ consists of i.i.d.\ random selectors with mean $m/N$, $I=\{i\in [N] \ : \ \theta_i=1\}$ and $g,h$ are independent $n$-dimensional standard Gaussian vectors that are independent of $\theta$.
\begin{theorem}
\label{thm:mainNT}
Fix $\del>0$. Let $A=R_I\Gamma_g$ be a randomly subsampled Gaussian circulant matrix and let $B=R_I[\Gamma_g \ h]$. 
Under the assumptions on $s,m,N,\delta, \eta$ of Theorem~\ref{thm:main1}, $\frac{1}{m}\sqrt{\frac{\pi}{2}} A$ satisfies RIP$_{1,2}(s,\del)$ with probability at least $1-\eta$.  
Moreover, 
under the assumptions of Theorem~\ref{thm:main2} $\frac{1}{m}\sqrt{\frac{\pi}{2}} A$ and $\frac{1}{m}\sqrt{\frac{\pi}{2}}B$ satisfy RIP$_{1,2}^{\text{eff}}(s,\del)$ with probability at least $1-\eta$.
\end{theorem}
\begin{proof}
Let $\kappa > 0$ be a number to be chosen later. Let $\cN_{\del/(1+\kappa)}\subset \Si_{s,N}$ be a minimal $\del/(1+\kappa)$-net for $\Si_{s,N}$ with respect to the Euclidean norm. It is well known, see e.g.\ \cite[Proposition C.3]{fora13}, that 
$$\log |\cN_{\del/(1+\kappa)}| \leq s\log\left(\frac{3(1+\kappa) eN}{s\del}\right).$$
Fix $x\in \Si_{s,N}$ and let $y\in \cN_{\del/(1+\kappa)}$ be such that $\|x-y\|_2\leq \del/(1+\kappa)$. 
We consider the events
\begin{equation}
\label{eqn:events}
\begin{split}
E_I & = \Big\{\frac{m}{2}\leq |I|\leq \frac{3m}{2}\Big\}  \\
E_{\text{RIP}} & =  \Big\{ \forall \ z \in \Si_{2s,N} \ : \ 
\frac{1}{\sqrt{m}}\|Az\|_2\leq 1+\kappa \Big\}   \\
E_{\Gamma,\ell_1} & = \Big\{\forall y\in \cN_{\del/(1+\kappa)} \ : \ \Big| \frac{1}{N}\sqrt{\frac{\pi}{2}} \|\Ga_g y\|_1 - 1\Big| \leq \del\Big\} \\
E & = \Big\{\forall y\in \cN_{\del/(1+\kappa)} \ : \  \Big|\frac{1}{m}\sqrt{\frac{\pi}{2}} \|Ay\|_1 - \frac{1}{N}\sqrt{\frac{\pi}{2}} \|\Ga_g y\|_1\Big| \leq 2\del\Big\} 
\end{split}
\end{equation}
respectively. Under $E_I$ and $E_{\text{RIP}}$, 
\begin{align*}
\Big|\frac{1}{m}\sqrt{\frac{\pi}{2}} \|Ax\|_1 - \frac{1}{m}\sqrt{\frac{\pi}{2}} 
\|Ay\|_1\Big| 
& \leq \frac{\del}{1+\kappa} \frac{1}{m}\sqrt{\frac{\pi}{2}} 
\Big\|A\Big(\frac{x-y}{\|x-y\|_2}\Big)\Big\|_1 \\
& \leq \frac{\del}{1+\kappa} \frac{|I|}{m}  \sup_{z\in \Si_{2s,N}} 
\frac{1}{\sqrt{|I|}}\sqrt{\frac{\pi}{2}} \|Az\|_2 \leq 
\sqrt{\frac{9\pi}{8}}\del.
\end{align*}
Therefore, if the events in \eqref{eqn:events} hold simultaneously, then by the triangle inequality 
\begin{align*}
& \Big|\frac{1}{m}\sqrt{\frac{\pi}{2}} \|Ax\|_1 - 1\Big| \nonumber\\
&  \qquad \leq \Big|\frac{1}{m}\sqrt{\frac{\pi}{2}} \|Ax\|_1 - \frac{1}{m}\sqrt{\frac{\pi}{2}} \|Ay\|_1\Big| + \Big|\frac{1}{m}\sqrt{\frac{\pi}{2}} \|Ay\|_1 - \frac{1}{N}\sqrt{\frac{\pi}{2}} \|\Ga_g y\|_1\Big| \\
& \qquad \qquad \qquad \qquad \qquad +  \Big| \frac{1}{N}\sqrt{\frac{\pi}{2}} \|\Ga_g y\|_1 - 1\Big|\leq (\sqrt{9\pi/8}+3)\del. \nonumber
\end{align*}
Hence, it remains to show that the events in \eqref{eqn:events} hold with probability at least $1-\eta$. The Chernoff bound immediately yields $\bP(E_I^c)\leq e^{-cm}$. 
By Theorem~\ref{thm:RIP2upper}, under the event $E_I$, if
\begin{equation}\label{m:cond-kappa}
m \gtrsim \kappa^{-2}s(\log^2(s)\log^2(N) + \log(1/\eta))
\end{equation}
then
\[
\bP_g(E_{\text{RIP}}^c)  = \bP_g\Big(\exists z\in \Si_{2s,N} \ : \ \frac{1}{\sqrt{|I|}}\|R_I\Gamma_g z\|_2 \geq 1+\kappa) \Big)
 \leq \eta
 \]
Moreover, by \eqref{eqn:Gammal1} and a union bound,
\begin{align*}
\bP\Big(\sup_{y\in \cN_{\del/(1+\kappa)}}\Big|\frac{1}{N} \sqrt{\frac{\pi}{2}} \|\Ga_g y\|_1 - 1\Big|\geq \del\Big) & \leq |N_{\del/(1+\kappa)}| 2e^{-\del^2 N/\pi s} \\
&\leq 2e^{s\log(3e(1+\kappa)N/(s\del)) - \del^2 N/\pi s}, 
\end{align*}
so $\bP(E_{\Gamma,\ell_1}^c) \leq 2\eta$ if 
\begin{equation}\label{N:cond}
N\geq \pi \del^{-2}s^2 \log(3eN(1+\kappa)/s\del) + \pi \del^{-2} s\log(1/\eta).
\end{equation}
Thus, it remains to show that $\bP(E^c)\leq \eta$.  
To prove this, we consider 
\[
X_y=\frac{1}{m}\sqrt{\frac{\pi}{2}} \|D_{\theta} \Ga_g y\|_1 - \frac{1}{N}\sqrt{\frac{\pi}{2}} \|\Ga_g y\|_1
\]
and $X_y'=\frac{1}{m}\sqrt{\frac{\pi}{2}}\|D_{\theta'} \Ga_g y\|_1 - 
\frac{1}{N}\sqrt{\frac{\pi}{2}} \|\Ga_g y\|_1$ for $y\in \cN_{\del}$, where 
$\theta'$ is an independent copy of $\theta$. By symmetrization, see 
Lemma~\ref{lem:symmetr},
\begin{align*}
\bP_{\theta}\Big(\sup_{y\in \cN_{\del/(1+\kappa)}} |X_y|\geq 2\del\Big) \leq \bP_{\theta}\Big(\sup_{y\in \cN_{\del/(1+\kappa)}} |X_y-X_y'|\geq \del\Big) 
 + \sup_{y\in \cN_{\del/(1+\kappa)}} \bP_{\theta}(|X_y|\geq \del)
\end{align*}
and so 
\begin{align}
\label{eqn:symm}
&\bP_{\theta,g}\Big(\sup_{y\in \cN_{\del/(1+\kappa)}} |X_y|\geq 2\del\Big) \leq \bP_{\theta,g}\Big(\sup_{y\in \cN_{\del/(1+\kappa)}} |X_y-X_y'|\geq \del\Big) \\
& \qquad\qquad + \E_g \sup_{y\in \cN_{\del/(1+\kappa)}} \bP_{\theta}(|X_y|\geq \del). \nonumber
\end{align}
To bound the first term on the right hand side, observe that $X_y-X_y'$ and  
$$\frac{1}{m}\sqrt{\frac{\pi}{2}} \sum_{i=1}^N \eps_i (\theta_i-\theta_i') 
|\langle a_i,y\rangle|,$$
where $\eps$ is a Rademacher vector, i.e., a vector of independent random signs, 
that is independent of $g,\theta$ and $\theta'$,  are identically distributed 
and so
\begin{align*}
& \bP_{\theta,\theta',g}\Big(\sup_{y\in \cN_{\del/(1+\kappa)}} |X_y-X_y'|\geq 
\del\Big) \\
& \qquad = \bP_{\eps,\theta,\theta',g}\Big(\frac{1}{m}\sqrt{\frac{\pi}{2}} 
\sup_{y\in \cN_{\del/(1+\kappa)}}\Big|\sum_{i=1}^N \eps_i (\theta_i-\theta_i') 
|\langle a_i,y\rangle|\Big| \geq \del\Big) \\
& \qquad \leq 2\bP_{\eps,\theta,g}\Big(\frac{1}{m}\sqrt{\frac{\pi}{2}} \sup_{y\in \cN_{\del/(1+\kappa)}}\Big|\sum_{i=1}^N \eps_i \theta_i |\langle a_i,y\rangle|\Big| \geq \del/2\Big)\\
& \qquad \leq 2 \bP_{\eps,\theta,g}\Big(\frac{1}{m}\sqrt{\frac{\pi}{2}} \sup_{y\in \cN_{\del/(1+\kappa)}}\Big|\sum_{i=1}^N \eps_i \theta_i 1_{E_{\text{RIP}}} |\langle a_i,y\rangle|\Big| \geq \del/2\Big) + \bP_{\theta,g}(E_{\text{RIP}}^c).
\end{align*}
By Hoeffding's inequality and assuming $E_I$, 
\begin{align*}
\bP_{\eps}\Big(\frac{1}{m}\sqrt{\frac{\pi}{2}}\Big|\sum_{i=1}^N \eps_i \theta_i 1_{E_{\text{RIP}}} |\langle a_i,y\rangle|\Big| \geq \del/2\Big) & \leq 2\exp\Big(-\frac{m^2 \del^2}{2\pi \sum_{i=1}^N \theta_i 1_{E_{\text{RIP}}} |\langle a_i,y\rangle|^2}\Big) \\
& = 2\exp\Big(-\frac{m \del^2}{2\pi1_{E_{\text{RIP}}} \frac{1}{m} \|Ay\|_2^2}\Big) \leq 2e^{-\frac{m\del^2}{2\pi(1+\kappa)^2}} 
\end{align*}
Hence a union bound 
yields
$$\bP_{\eps,\theta,g}\Big(\frac{1}{m}\sqrt{\frac{\pi}{2}} \sup_{y\in \cN_{\del/(1+\kappa)}}\Big|\sum_{i=1}^N \eps_i \theta_i 1_{E_{\text{RIP}}} |\langle a_i,y\rangle|\Big| \geq \del/2\Big) \leq 2|\cN_{\del/(1+\kappa)}|e^{-\frac{m\del^2}{2\pi(1+\kappa)^2}} \leq \eta$$
provided that
\begin{equation}\label{cond:m2}
m \gtrsim \frac{(1+\kappa)^2}{\delta^{2}} \left(s \log\left(\frac{3e(1+\kappa)N}{\delta s}\right) + \log(1/\eta)\right).  
\end{equation}
To bound the second term on the right hand side of \eqref{eqn:symm}, consider the event
$$E_{\Gamma,\ell_2} = \Big\{\forall y\in \cN_{\del/(1+\kappa)} \ : \ \frac{1}{\sqrt{N}} \|\Ga_g y\|_2 \leq 2\Big\}.$$
By \eqref{eqn:Gammal2} and a union bound, 
$$\bP_g(E_{\Gamma,\ell_2}^c) \leq 2 |\cN_{\del/(1+\kappa)}| e^{-cN/s} \leq \eta$$
under the condition $N \gtrsim \delta^{-2} s^2 \log\left( \frac{3e(1+\kappa) 
N}{s\delta}\right) + s \log(1/\eta)$, which is weaker than \eqref{N:cond}.
This shows that 
$$\E_g \sup_{y\in \cN_{\del/(1+\kappa)}} \bP_{\theta}(|X_y|\geq \del) \leq \E_g \sup_{y\in \cN_{\del/(1+\kappa)}} \bP_{\theta}(|X_y 1_{E_{\Gamma,\ell_2}}|\geq \del) + \eta.$$
Now recall the following facts. If $X$ is a random variable, $X'$ is an 
independent copy, and $\text{med}(X)$ is a median of $X$, then for any $\del>0$ 
(see Lemma \ref{lem:basic_facts})
$$\bP(|X-\operatorname{med}(X)|\geq \del) \leq 2\bP(|X-X'|\geq \del)$$
and 
$$|\operatorname{med}(X)-\E X| \leq (\E(X-\E X)^2)^{1/2}.$$
Combining these, we find
$$\bP(|X-\E X|\geq \del) \leq 2\bP(|X-X'|\geq \del - (\E(X-\E X)^2)^{1/2}).$$
We apply these inequalities with $X=X_y 1_{E_{\Gamma,\ell_2}}$ and $X'=X_y' 1_{E_{\Gamma,\ell_2}}$. Note that $\E_{\theta} X=0$.
By symmetrization, see e.g.\ \cite[Lemma 6.3]{LeT91} or \cite[Lemma 8.4]{fora13},
\begin{align*}
&(\E_{\theta}(X-\E_{\theta} X)^2)^{1/2}) 
= \Big(\E_{\theta} \Big|\frac{1}{m}\sqrt{\frac{\pi}{2}} \|Ay\|_1 - \E_{\theta}\Big(\frac{1}{m}\sqrt{\frac{\pi}{2}} \|Ay\|_1\Big)\Big|^2\Big)^{1/2}1_{E_{\Gamma,\ell_2}} \\
& \qquad 
\leq \frac{2}{m}\sqrt{\frac{\pi}{2}} \Big(\E_{\theta,\eps}\Big|\sum_{i=1}^N \eps_i \theta_i |\langle a_i,y\rangle| \Big|^2\Big)^{1/2}1_{E_{\Gamma,\ell_2}} \\
& \qquad 
= \frac{2}{m}\sqrt{\frac{\pi}{2}} \Big(\E_{\theta}\sum_{i=1}^N \theta_i |\langle a_i,y\rangle|^2\Big)^{1/2}1_{E_{\Gamma,\ell_2}}  
= \frac{1}{\sqrt{m}} \sqrt{2\pi} \frac{1}{\sqrt{N}} \|\Ga_g y\|_2 1_{E_{\Gamma,\ell_2}} \\
&\qquad \leq \frac{2\sqrt{2\pi}}{\sqrt{m}} \leq \del/2 
\end{align*}
if $m\geq 32\pi \del^{-2}$ the latter being a weaker condition than 
\eqref{cond:m2}. 
In summary, we find
\begin{align*}
\bP_{\theta}(|X_y 1_{E_{\Gamma,\ell_2}}|\geq \del) & \leq 2\bP(|X_y-X_y'|1_{E_{\Gamma,\ell_2}}\geq \del/2) \\
 & \leq 4\bP_{\theta,\eps}\Big(\frac{1}{m}\sqrt{\frac{\pi}{2}} \sum_{i=1}^N \eps_i \theta_i |\langle a_i,y\rangle|1_{E_{\Gamma,\ell_2}} \geq \frac{\del}{4}\Big).
\end{align*}
Now apply Hoeffding's inequality to obtain
\begin{align*}
\bP_{\theta}(|X_y 1_{E_{\Gamma,\ell_2}}|\geq \del) & \leq 4\E_{\theta}\exp\Big(\frac{-m^2\frac{2}{\pi}\del^2}{16\sum_{i=1}^N \theta_i |\langle a_i,y\rangle|^2}\Big) 
& = 4\E_{\theta}\exp\Big(\frac{-m \del^2}{8\pi \frac{1}{m} \|Ay\|_2^2}\Big) \\
& \leq 4e^{-\frac{m\del^2}{8\pi(1+\kappa)^2}} + 4\bP_{\theta}(E_{\text{RIP}}^c). 
\end{align*}
If $m\geq 8\pi \frac{(1+\kappa)^2}{\del^{2}}\log(1/\eta)$ (which is again weaker than \eqref{cond:m2}), we find
$$\E_g \sup_{y\in N_{\del}} \bP_{\theta}(|X_y|\geq \del)\leq 5\eta +4 \bP_{\theta,g}(E_{\text{RIP}}^c)\leq 9\eta.$$
We still need to choose $\kappa > 0$ and distinguish two cases to this end. 

{\bf Case 1:} Assume that $0 < \delta < \big(\log(s) \sqrt{\log(N)}\big)^{-1}$ and choose $\kappa = 1$.
A nontrivial $s \geq 1$ is only allowed by 
\eqref{s:cond} if $\delta \gtrsim 1/\sqrt{N}$. In this situation we have
$\log(3e(1+\kappa)N/(s\delta)) \asymp \log(N)$.
Then \eqref{s:cond} implies \eqref{N:cond}. Moreover,  \eqref{eqn:condSp:small} implies both \eqref{cond:m2} and \eqref{m:cond-kappa} for our conditions on the parameters $\kappa$ and $\delta$.
Altogether, we obtain $\bP(E^c)\leq 10\eta$ in this case.

{\bf Case 2:} If $\big(\log(s) \sqrt{\log(N)}\big)^{-1} < \delta \leq 1$ we choose $\kappa = \sqrt{\delta \log(s)} \log^{1/4}(N) > 1$.
Again a nontrivial $s \geq 1$ implies $\delta \gtrsim 1/\sqrt{N}$ by \eqref{s:cond} and also in this case we have
$\log(3e(1+\kappa)N/(s\delta)) \asymp \log(N)$. Plugging our choice of $\kappa$ into \eqref{cond:m2} and 
\eqref{m:cond-kappa}, we observe that both these conditions are implied by \eqref{eqn:condSp:large} and $\bP(E^c)\leq 10\eta$.\par

\medskip

The proof of the second statement for $A$ is similar, so we only indicate the necessary changes in the argument. Let us write $C_{s,N} = \{x\in \R^N \ : \ \|x\|_1\leq \sqrt{s}, \ \|x\|_2 = 1\}$. It clearly suffices to show that
$$\sup_{x\in C_{s,N}} \Big|\frac{1}{m}\sqrt{\frac{\pi}{2}} \|Ax\|_1 - 1\Big|\leq \del$$
with probability at least $1-\eta$. Let us first recall that $C_{s,N}\subset 2\text{conv}(\Si_{s,N})$ \cite[Lemma 3.1]{PlV13lin}. Hence, under $E_{\text{RIP}}$, 
$$\frac{1}{\sqrt{m}}\sup_{z\in C_{s,N}}\|Az\|_2 \leq 2\frac{1}{\sqrt{m}}\sup_{z\in \Si_{s,N}}\|Az\|_2\leq 2(1+\kappa).$$
We repeat the above argument with $\cN_{\del/(1+\kappa)}$ replaced by a minimal $\del/(1+\kappa)$-net of $C_{s,N}$. Using $C_{s,N}\subset 2\text{conv}(\Si_{s,N})$ and Sudakov's inequality, Theorem~\ref{thm:Sudakov}, 
we find 
\begin{align}
\label{eqn:appSud}
\log|\cN_{\del/(1+\kappa)}| & \lesssim \frac{(1+\kappa)^2}{\del^{2}} \Big(\E\sup_{x\in C_{s,N}}\langle g,x\rangle\Big)^2 \leq \frac{4(1+\kappa)^2}{\del^{2}} \Big(\E\sup_{x\in \Si_{s,N}}\langle g,x\rangle\Big)^2\\
& \lesssim \frac{4(1+\kappa)^2}{\del^{2}} s\log(eN/s),
\end{align}
where the final inequality is \cite[Lemma 2.3]{PlV13}. By now chasing through the argument above we arrive at the three conditions
\begin{equation}\label{mN:strengthened}
\begin{split}
 N & \gtrsim \frac{(1+\kappa)^2}{\delta^4} s^2 \log(eN/s) + s 
\frac{1}{\delta^2}  \log(1/\eta)\\
 m & \gtrsim \kappa^{-2} s ( \log^2(s) \log^2(N) + \log(1/\eta)) \\
 m & \gtrsim \frac{(1+\kappa)^4}{\delta^4} s \log(eN/s) + \frac{(1+\kappa)^2}{\delta^2} \log(1/\eta).
 \end{split}
\end{equation}
Again, we distinguish two cases depending on $\delta$ and choose $\kappa$ as
\[
\kappa = \left\{ \begin{array}{ll} 1 & \mbox{ if } 0 < \delta \leq (\log^2(s) \log(N) )^{-1/4}, \\
( \delta^4 \log^2(s) \log(N))^{1/6} & \mbox{ if } (\log^2(s) \log(N) )^{-1/4} < \delta \leq 1.
\end{array} \right.
\]
With this we can deduce the statement of the theorem (noting also that $\log(s) \leq \log(N)$).\par

Finally, let us prove the second statement for $B$. Let $N_{\del/(1+\kappa)}$ be a minimal $\del/(1+\kappa)$-net of $C_{s,N+1}$. By the first part of the proof, it is readily seen that the result will follow once we show that the events
\begin{equation}
\label{eqn:eventsGaExt}
\begin{split}
E_{\text{RIP},B} & =  \Big\{ \forall \ z \in \Si_{2s,N+1} \ : \ 
\frac{1}{\sqrt{m}}\|Bz\|_2\leq 3+\kappa\Big\}   \\
E_{\Gamma,h,\ell_2} & = \Big \{\forall y\in N_{\del} \ : \ \frac{1}{\sqrt{N}} 
\|[\Ga_g \ h] y\|_2 \leq 4 \Big\}\\
E_{\Gamma,h,\ell_1} & = \Big\{\forall y\in N_{\del} \ : \ \Big| 
\frac{1}{N}\sqrt{\frac{\pi}{2}} \|[\Ga_g \ h] y\|_1 - 1\Big| \leq \del \Big\} 
\end{split}
\end{equation}
hold with probability at least $1-c \eta$. For $E_{\Gamma,h,\ell_1}$ this is immediate from \eqref{eqn:Gammal1app} and a union bound. For $E_{\text{RIP},B}$, observe that 
$$\frac{1}{\sqrt{m}}\|Bz\|_2 \leq \frac{1}{\sqrt{m}} \|Az_{[N]}\|_2 + 
|z_{N+1}| \ \frac{1}{\sqrt{m}}\sqrt{\frac{\pi}{2}}\|D_{\theta} h\|_2.$$
We have already seen that the event $E_I=\{\frac{m}{2}\leq |I| \leq \frac{3m}{2}\}$ holds with probability $1-\eta$. Under this event, the Hanson-Wright inequality \eqref{eqn:H-W} yields
$$\bP_g\Big(\frac{1}{m}\|D_{\theta}h\|_2^2 \geq 2\Big) \leq 
\bP_g\Big(\frac{1}{|I|}\|D_{\theta}h\|_2^2 \geq \frac{4}{3}\Big) \leq 
e^{-c|I|}\leq e^{-m/2}\leq \eta$$
for $m\gtrsim \log(1/\eta)$. Under the event $E_{RIP}$ we have
$$\frac{1}{\sqrt{m}} \sup_{z \in \Si_{2s,N}}\|Az_{[N]}\|_2 \leq 1+\kappa$$
with probability $1-\eta$, so that, with probability at least $1-2\eta$,
$$\frac{1}{\sqrt{m}}\|Bz\|_2 \leq (1+\kappa) \| z_{[N]}\|_2 + 2 |z_{N+1}| \leq 
(3+\kappa) \|z\|_2$$
under the conditions \eqref{mN:strengthened}. 
Very similarly, one shows that $E_{\Gamma,h,\ell_2}$ holds with probability at least $1-\eta$ under \eqref{mN:strengthened}.
As before, distinguishing two cases for $\delta$ one arrives at the statement of the theorem.  
\end{proof}

\section{Further applications}
\label{sec:appl}

Apart from its usefulness for one-bit compressed sensing, the RIP$_{1,2}$-property is of interest for (unquantized) outlier robust compressed sensing \cite{DLR16} and for compressed sensing involving uniformly scalar quantized measurements \cite{JaC16,MJC16}. In this section, we briefly sketch the implications of Theorem~\ref{thm:mainNT} for these two directions. 
\begin{corollary}
\label{cor:BPDN1}
Let $A=R_I\Gamma_g$ be a randomly subsampled Gaussian circulant matrix. Let $0 < \eta < 1$ and $s \in [N]$ such that
\[
s \lesssim \min\left\{\sqrt{N/\log^2(N)}, N/\log(1/\eta)\right\}
\]
and suppose that
\begin{equation}\label{m:l1}
m \gtrsim s \max\left\{\log^{4/3}(s) \log^{5/3}(N), \frac{\log(1/\eta)}{\log^{3/2}(s) \log^{1/3}(N)}, \frac{\log^{2/3}(s) \log^{1/3}(N)}{s}\right\}.
\end{equation}
Then, with probability exceeding $1-\eta$ the following holds: for any $x\in \C^n$ and $y=Ax+e$, where $\|e\|_1\leq \eps$, any solution $x^{\#}$ to
\begin{equation*}
\min_{z\in \C^n} \|z\|_1 \quad \mbox{ subject to } \quad
\|y-Az\|_1\leq \eps.
\end{equation*}
satisfies
$$\|x-x^{\#}\|_2\lesssim \frac{\si_s(x)_1}{\sqrt{s}} + \frac{\eps}{m}.$$
\end{corollary}
\begin{proof}
As is argued in the proof of \cite[Theorem III.3]{DLR16}, it suffices to show that with probability at least $1-\eta$,  
\begin{equation}
\frac{1}{m}\|Ax\|_1 \geq c\|x\|_2,\qquad \text{for all} \ x\in \Si_{s,N}^{\text{eff}}.
\end{equation}
for a universal constant $c>0$. Hence the result immediately follows from Theorem~\ref{thm:mainNT} by choosing 
$\delta = c$ constant.
\end{proof}
Note that after estimating $\log(s) \leq \log(N)$ and for inverse polynomial probability $\eta = N^{-2}$, say, condition \eqref{m:l1} takes
the simpler form
\[
m \gtrsim s \log^3(N).
\] 

In addition, we can use Theorem~\ref{thm:mainNT} to derive the following reconstruction result involving a uniform scalar quantizer that uses additional dithering. Let $Q_{\del}:\R^m\rightarrow (\del\Z + \del/2)^m$ be the uniform scalar quantizer with resolution $\del$ defined by $Q_{\del}(z) = \big(\del\lfloor z_i/\del\rfloor + \del/2\big)_{i=1}^m$. 
\begin{theorem}
\label{thm:USC}
Let $A=R_I\Gamma_g$ be a randomly subsampled Gaussian circulant matrix. Let $\tau$ be a vector of $m$ independent $\mathcal{N}(0,\pi R^2/2)$-distributed random variables. Suppose that $u$ is a vector of $m$ independent random variables that are uniformly distributed on $[0,\del]$ and are independent of $\tau$. Assume that, for $0<\eta, \epsilon \leq1$,
\begin{equation}\label{eqn:condUSC}
\begin{split}
s &\lesssim \min\left\{\sqrt{N/\log^2(N)}, N/\log(1/\eta)\right\} \\
m &\gtrsim s \max\left\{\log^{4/3}(s) \log^{5/3}(N), \frac{\log(1/\eta)}{\log^{3/2}(s) \log^{1/3}(N)}, \frac{\log^{2/3}(s) \log^{1/3}(N)}{s},\right.\\
& \phantom{\gtrsim s \max\big\{\;} \left. R^2 \delta^{-2} \epsilon^{-6} 
\log(1/\eta), \frac{\log(1/\eta)}{\epsilon^2 s}\right\}.
\end{split}  
\end{equation}
Then, with probability at least $1-\eta$ the following holds: for any $x\in \R^N$ with $\|x\|_1\leq \sqrt{s}\|x\|_2$ and $\|x\|_2\leq R$, any solution $x^{\#}$ to the program
\begin{equation}
\label{eqn:CPUSC}
\min\|z\|_1 \ \ \ \operatorname{s.t.} \ \ \ \|z\|_2\leq R, \ Q_{\del}\Big(\sqrt{\frac{\pi}{2}} Az + \tau + u\Big) = Q_{\del}\Big(\sqrt{\frac{\pi}{2}} Ax + \tau + u\Big)
\end{equation}
satisfies $\|x^{\#}-x\|_2\lesssim \del \epsilon$.
\end{theorem}
Note that the program \eqref{eqn:CPUSC} is convex. Indeed, the second condition in \eqref{eqn:CPUSC} is equivalent to 
$$\frac{\del}{2}\leq \Big(\sqrt{\frac{\pi}{2}} Az + \tau + u\Big)_i - Q_{\del}\Big(\sqrt{\frac{\pi}{2}} Ax + \tau + u\Big)_i <\frac{\del}{2}, \qquad i=1,\ldots,m.$$ 
\begin{proof}
Let $x^{\#}$ be any solution to \eqref{eqn:CPUSC}. Since $x$ is feasible for \eqref{eqn:CPUSC}, 
$$\|[x^\#,R]\|_1\leq\|[x,R]\|_1\leq \sqrt{s}\|x\|_2 + R\leq R(\sqrt{s} + 1).$$ 
Since $R\leq \|[x^{\#},R]\|_2,\|[x,R]\|_2\leq \sqrt{2}R$, it follows that 
$$[x^{\#},R],[x,R] \in \Si^{\operatorname{eff}}_{(\sqrt{s}+1)^2,N+1}\cap \sqrt{2}R \ B_{\ell_2^{N+1}}.$$ 
Moreover, the last condition in \eqref{eqn:CPUSC} precisely means that 
$$Q_{\del}\Big(\sqrt{\frac{\pi}{2}} B[x^{\#},R] + u\Big) = Q_{\del}\Big(\sqrt{\frac{\pi}{2}} B[x,R] + u\Big),$$
where $B=R_I[\Gamma_g \ h]$, where $h$ is an independent standard Gaussian that is independent of $g$ and $\theta$. We will now show that this implies the reconstruction error bound.\par 
Under \eqref{eqn:condUSC}, Theorem~\ref{thm:mainNT} implies that 
$$\frac{1}{2}\|z\|_2\leq \frac{1}{m}\sqrt{\frac{\pi}{2}} \|Bz\|_1\leq \frac{3}{2}\|z\|_2 \qquad \mbox{ for all } z\in \Si^{\operatorname{eff}}_{(\sqrt{s}+1)^2,N+1}$$ 
with probability at least $1-\eta$. By \cite[Proposition 1]{JaC16} we obtain that under this event, $\sqrt{\frac{\pi}{2}} B$ satisfies with probability at least $1-\eta$ a quantized version of RIP$_{1,2}$:
for some universal constant $c$ and any $z,z'\in \Si^{\operatorname{eff}}_{(\sqrt{s}+1)^2,N+1}\cap \sqrt{2}R \ B_{\ell_2^{N+1}}$ 
\begin{equation}
\label{eqn:QRIP}
\frac{1}{2}\|z-z'\|_2 - c\del\epsilon \leq \Big\|Q_{\del}\Big(\sqrt{\frac{\pi}{2}} Bz+u\Big) - Q_{\del}\Big(\sqrt{\frac{\pi}{2}} Bz'+u\Big)\Big\|_1 \leq \frac{1}{2}\|z-z'\|_2 + c\del\epsilon
\end{equation}
provided that 
\begin{equation}
\label{eqn:condUSCm}
m\gtrsim \epsilon^{-2} \cN(\Si^{\operatorname{eff}}_{(\sqrt{s}+1)^2,N+1}\cap \sqrt{2}R \ B_{\ell_2^{N+1}},\|\cdot\|_2,\del\epsilon^2) + \epsilon^{-2}\log(1/\eta),
\end{equation}
where $\cN(\cdot,\|\cdot\|_2,t)$ denotes the covering number with respect to the Euclidean norm. Now observe that
\begin{align*}
\Si^{\operatorname{eff}}_{(\sqrt{s}+1)^2,N+1}\cap \sqrt{2}R \ B_{\ell_2^{N+1}} & \subset \sqrt{2}R \{x\in \R^{N+1} \ : \ \|x\|_1\leq\sqrt{2s}, \ \|x\|_2\leq 1\} \\
& \subset 2\sqrt{2}R \operatorname{conv}(\Si_{2s,{N+1}}),
\end{align*}
where the final inclusion holds by \cite[Lemma 3.1]{PlV13lin}. Therefore, similarly to \eqref{eqn:appSud}, Sudakov's inequality (Theorem~\ref{thm:Sudakov}) implies that
$$ \log\cN(\Si^{\operatorname{eff}}_{(\sqrt{s}+1)^2,N+1}\cap \sqrt{2}R \ B_{\ell_2^{N+1}},\|\cdot\|_2,\del\epsilon^2) \lesssim R^2\del^{-2}\epsilon^{-4}s\log(eN/s)$$ 
and it follows that \eqref{eqn:condUSCm} is satisfied under our assumptions \eqref{eqn:condUSC}. We can now apply \eqref{eqn:QRIP} with $z=[x,1]$ and $z'=[x^{\#},1]$ and use $\|z-z'\|_2=\|x-x^{\#}\|_2$ to obtain the asserted error bound.    
\end{proof}

\appendix 

\section{Uniform scalar quantization and recovery via $\ell_\infty$-constrained $\ell_1$-minimization}
\label{app:scalar-quant}
 
Let $Q_\delta : \R^m (\del\Z + \del/2)^m$ be the 
uniform scalar quantizer of resolution $\delta$ defined in Section~\ref{sec:appl}.
Quantized measurements take the form $y = Q_\delta(A x)$. A reconstruction procedure delivering $x^\sharp$ 
should be quantization consistent, i.e., $Q_\delta(Ax^\sharp) = Q_\delta(A x)$.  
This motivates to consider an $\ell_1$-minimization problem of the following 
type for the reconstruction.
Let 
\[
B_\delta = \{z \in \R^m : - \delta/2 \leq z_i \leq \delta/2, i=1,\hdots,m\}
\]
and let $x^\sharp$ be the minimizer of 
\begin{equation}\label{basispursuit-inf}
\min \|z\|_1 \mbox{ subject to } Az - y \in B_\delta.
\end{equation}
Then $Q_\delta(A x^\sharp) = y = Q_\delta(A x)$. Note that taking the closure of 
$B_\delta$ instead of $B_\delta$ the above problem becomes an 
$\ell_\infty$-constrained $\ell_1$-minimization problem.

The following result concerning reconstruction via \eqref{basispursuit-inf} based on the standard RIP 
apparently has not been observed before (but see the discussion in \cite{DLR16} for non-optimal results based on $\ell_p$ versions of the RIP for $p\neq 2$).
\begin{theorem}\label{thm:RIP-linfty} Suppose that $A \in \R^{m \times N}$ is such that $\frac{1}{\sqrt{m}}A$ 
satisfies the $\ell_2$-RIP \eqref{def:RIP2} for 
$\delta_{2s} < 4/\sqrt{41} \approx 0.62$. Then for any $x \in \R^N$ and $y = Q_\delta(Ax)$ a solution $x^\sharp$ of \eqref{basispursuit-inf} is quantization consistent and satisfies
\begin{equation}\label{quant:err:bound}
\|x - x^\sharp\|_\infty \lesssim \delta + s^{-1/2} \inf_{w \in \R^N, \|w\|_0 \leq s} \|x-w\|_1.
\end{equation}
\end{theorem} 
\begin{proof} The optimization problem \eqref{basispursuit-inf} is closely related to the $\ell_\infty$-constrained
$\ell_1$-minimization problem
\begin{equation}\label{bp_infty}
\min \|z\|_1 \mbox{ subject to } \|Az - y\|_\infty \leq \delta/2.
\end{equation}
In fact, either a minimizer of \eqref{basispursuit-inf} exists in which case it 
is also a minimizer of \eqref{bp_infty} or no minimizer of 
\eqref{basispursuit-inf} exists in which case the theorem is void. 
(A minimizer of \eqref{bp_infty} always exists so that it may be preferred in 
practice. The error bound \eqref{quant:err:bound} still holds for \eqref{bp_infty}, but every 
minimizer of \eqref{bp_infty} is quantization inconsistent in the case that no minimizer of \eqref{basispursuit-inf} exists.)

This close relation of \eqref{basispursuit-inf} and \eqref{bp_infty} suggests to study versions of the null space property, see e.g.\ \cite[Chapter 4]{fora13}. 
By Theorem \cite[Theorem 6.13]{fora13}, the bound on the restricted isometry constants of $\frac{1}{\sqrt{m}}$ 
implies the $\ell_2$-robust null space property in the form
\[
\|v_{S}\|_2 \leq \frac{\rho}{\sqrt{s}} \|v_{S^c}\|_1 + \tau \frac{1}{\sqrt{m}} \|A v\|_2 \quad \mbox{for all } v \in \R^N \mbox{ and all } S \subset [N], \#S =s,
\]
for constants $\rho \in (0,1)$ and $\tau > 0$ that only depend on $\delta_{2s}$.
Since $\|Av\|_2 \leq \sqrt{m} \|A v\|_\infty$ this yields
\[
\|v_{S}\|_2 \leq \frac{\rho}{\sqrt{s}} \|v_{S^c}\|_1 + \tau \|A v\|_\infty \quad \mbox{for all } v \in \R^N \mbox{ and all } S \subset [N], \#S =s,
\]
which is the $\ell_\infty$-robust null space property of order $s$. By \cite[Theorem 4.12]{fora13} this implies the error bound
for any minimizer $x^*$ of \eqref{bp_infty}
\[
\|x-x^*\|_2 \lesssim \frac{\inf_{w \in \R^N, \|w\|_0 \leq s} \|x-w\|_1}{\sqrt{s}} + \delta.
\]
This concludes the proof.
\end{proof}

By \cite{KMR14}, see also Section~\ref{sec:relatedwork}, a partial random circulant matrix $\frac{1}{\sqrt{|I|}} R_I \Gamma_g$ with subsampling on a fixed (deterministic) set $I \subset [N]$ generated by a standard Gaussian random vector 
satisfies $\delta_{2s} \leq 0.6$ with probability at least $1-\eta$ if
\[
|I| \gtrsim s (\log^2(s) \log^2(N) + \log(1/\eta)).
\]
Therefore, Theorem~\ref{thm:RIP-linfty} implies stable reconstruction from quantized measurements $y = Q_\delta(R_I \Gamma_g)$ via \eqref{basispursuit-inf} under this condition on the number of measurements.
 
\section{Upper RIP bound for partial random circulant matrices}

The proof of our main result only requires the upper $(\ell_2)$-RIP bound in \eqref{def:RIP2}. 
For (deterministic) subsampling of random circulant matrices we can deduce the following bound on $m$ from \cite{KMR14},
valid also for values of $\kappa \geq 1$.

\begin{theorem}\label{thm:RIP2upper} For a fixed (deterministic) subset $I 
\subset [N]$, let $A = R_I \Gamma_g$ be a draw of a random circulant matrix 
generated by a standard Gaussian vector $g$  and $\kappa > 0$.
If
\[
|I| \gtrsim \kappa^{-2} s (\log^2(s) \log(N) + \log(1/\eta))
\]
then $\sup_{x \in \Sigma_{s,N}} \frac{1}{\sqrt{m}} \| A x\|_2 \leq 1+\kappa$ with probability at least $1-\eta$.
\end{theorem}
\begin{proof} As argued in \cite{KMR14}, we can write $\frac{1}{\sqrt{|I|}} Ax = 
V_x g$ with $V_x = \frac{1}{\sqrt{|I|}} R_I \Gamma_x$.
Denote, $\mathcal{A}_{s,N} = \{V_x : x \in 
\Sigma_{s,N}\}$.
It follows then from \cite[Theorem 3.5(a)]{KMR14} that for every $p \geq 1$
\begin{equation}\label{moment:b}
\left(\E \sup_{x \in \Sigma_{s,N}} \|V_x  
g\|_2^p\right)^{1/p} \lesssim \gamma_2(\mathcal{A}_{s,N}, \|\cdot\|_{\ell_2 \to 
\ell_2}) + d_F(\mathcal{A}_{s,N}) + \sqrt{p} d_{\ell_2 \to 
\ell_2}(\mathcal{A}_{s,N}),
\end{equation}
where $\gamma_2(\mathcal{A}_{s,N}, \|\cdot\|_{\ell_2 \to \ell_2})$ denotes the 
$\gamma_2$-functional of the set $\mathcal{A}_{s,N}$ with respect to the 
spectral norm, $d_{\ell_2 \to \ell_2}$ and $d_{F}$ denote the diameter in the 
spectral and Frobenius norm, respecticely, of the set in the argument, see 
\cite{KMR14} for details. These parameters have been estimated in \cite[Section 
4]{KMR14},
\begin{align*}
d_F(\mathcal{A}_{s,N}) & = 1, \qquad d_{\ell_2 \to \ell_2}(\mathcal{A}_{s,N}) 
\leq \sqrt{s/|I|} \\
\gamma_2(\mathcal{A}_{s,N},\|\cdot\|_{\ell_2 \to \ell_2}) & \lesssim 
\sqrt{s/|I|} \log(s) \log(N).
\end{align*}
Moreover, the moment bound \eqref{moment:b} implies a tail bound (see e.g.\ \cite[Prop.\ 2.6]{KMR14}), so that
\[
\bP\left(\sup_{x \in \Sigma_{s,N}} \|A x\|_2 \geq c(1+ \sqrt{s/|I|} \log(s) 
\log(N)) + t\right) 
\leq e^{- c\frac{|I| t^2}{s}}.
\]
Requiring that the right hand is bounded by $\eta$ gives the statement of the theorem.
\end{proof}

\section{Some tools from probability}

The Sudakov inequality provides a bound of the covering numbers in terms of the Gaussian widths, see e.g.\ \cite[Theorem 3.18]{LeT91}.

\begin{theorem}[Sudakov]\label{thm:Sudakov} Let $T \subset \R^N$ and $\delta > 0$. Then the covering numbers with respect to the Euclidean norm obey
\[
\log N(T,\|\cdot\|_2,\delta) \lesssim \delta^{-2} \left( \E \sup_{x \in T} \langle x, g \rangle \right)^2,
\]
where $g$ is a standard Gaussian random vector in $\R^N$.
\end{theorem}

\begin{lemma}\label{lem:basic_facts}
 Let $X$ denote a random variable  and let $\operatorname{med}(X)$ denote a 
median of $X$. Suppose that $\E(X-\E X)^2$ is finite and let $X'$ 
denote an independent copy of $X$. Then the 
following inequalities hold true:
 \begin{align*}
  \bP(|X - \operatorname{med}(X)| \geq \delta) &\leq 2 \bP(|X-X'| \geq \delta) \\
  |\E X - \ \operatorname{med}(X)| &\leq (\E(X-\E X)^2)^{\frac{1}{2}}
 \end{align*}
\end{lemma}
The estimates in Lemma \ref{lem:basic_facts} are well known. We provide a proof 
for convenience.
\begin{proof}
To show the first inequality, observe that since $X$ and $X'$ are identically
distributed the following holds for every $\delta>0$
 \begin{equation*}
  \bP(X' \leq \operatorname{med}(X)) \bP(X  \geq  \operatorname{med}(X) + 
\delta) \leq  \bP(X  \geq X' +  \delta)  \; .
 \end{equation*}
Since $\bP(X' \leq \operatorname{med}(X) )\geq \frac{1}{2}$ we conclude that
$
 \bP(X  \geq \operatorname{med}(X) +  \delta) \leq 2  \bP(X \geq X'+  \delta) 
$
implying the estimate 
\begin{equation*}
\bP(X -  \operatorname{med}(X) \geq  \delta) \leq 2  
\bP(X - X' \geq  \delta) \; .
\end{equation*}
Repeating the argument with $\bP(X-\operatorname{med}(X)  \leq - \delta)$ and $\bP(X'
\geq \operatorname{med}(X) )\geq \frac{1}{2}$ we arrive at the estimate
\begin{equation*}
 \bP(X -  \operatorname{med}(X) \leq -\delta) \leq 2  
\bP(X - X' \leq  -\delta) \; .
\end{equation*}
Since the events $\{X-X' \leq - \delta \}$, $\{X -  X' \geq  \delta\}$ 
and respectively the events $\{X-\operatorname{med}(X) \leq - \delta \}$, $\{X 
-  \operatorname{med}(X) \geq  \delta\}$ are disjoint, it follows that
\begin{align*}
 \bP(|X-\operatorname{med}(X)| \geq \delta) 
 &= \bP(X-\operatorname{med}(X) \leq - \delta ) + \bP(X-\operatorname{med}(X) 
\geq\delta ) \\
 &\leq 2\bP(X-X' \leq - \delta ) +2 \bP(X-X' \geq \delta ) \\
 &= 2\bP(|X- X'| \geq \delta) \; .
\end{align*}
To show the second inequality observe that with
$\sigma = (\E(X-\E X)^2)^{1/2}$ Cantelli's inequality yields
\begin{equation*}
\bP(X \geq \E X + 
\sigma) \leq \frac{\E(X-\E X)^2}{\sigma^2 + \sigma^2} = \frac{1}{2}  \; .
\end{equation*}
By using Cantelli's inequality with $\sigma$ replaced by $-\sigma$
we obtain the estimate
\begin{equation*}
\bP(X \geq \E X - 
\sigma) \geq 1 - \frac{\E(X-\E X)^2}{\sigma^2 + \sigma^2} = \frac{1}{2} \; .
\end{equation*}
Therefore $\bP(X < \E X - \sigma) \leq \frac{1}{2}$ and  $\bP(X \geq \E X + 
\sigma) \leq \frac{1}{2}$, which shows that every median of $X$ must
satisfy
\begin{equation*}
  \E X - \sigma \leq \operatorname{med}(X) \leq \E X + \sigma \; .
\end{equation*}
This implies the second inequality and completes the proof.
\end{proof}
The following probability bound related to symmetrization follows in the same way as in \cite[eq.\ (6.3)]{LeT91}.
\begin{lemma}\label{lem:symmetr}
Let $(X_t)$, $t \in T$, be a family of random variables, indexed by a finite or countable set $T$, and let
$(X_t')$ be an independent copy of $(X_t)$. Then, for $x,y >0$, 
\[
\bP(\sup_{t \in T}|X_t| \geq x + y) \leq \bP(\sup_{t \in T}|X_t - X_t'| \geq x) + \sup_{t \in T} \bP(|X_t| \geq y). 
\]
\end{lemma}

\section*{Acknowledgements}
The first-named author would like to thank Laurent Jacques, Christian K\"ummerle, and Rayan Saab for stimulating discussions. All authors acknowledge funding from the DFG through the project Quantized Compressive Spectrum Sensing (QuaCoSS) which is part of the Priority Program SPP 1798 Compressive Sensing in Information Processing (COSIP). HJ and HR acknowledge funding by the German Israel Foundation (GIF) through the project
Analysis of Structured Random Matrices in Recovery Problems (G-1266).


\end{document}